\documentclass[aps,prd,preprint,superscriptaddress,longbibliography,nofootinbib,floatfix]{revtex4-2}
\usepackage{CJKutf8}
\usepackage{booktabs}
\usepackage{graphicx}
\usepackage{amsmath,amssymb,amsthm}
\usepackage{mathrsfs}
\usepackage{orcidlink}
\numberwithin{equation}{section}

\theoremstyle{definition}

\theoremstyle{remark}

\begin{document}

\title{Open-channel quantization of JT gravity at finite cutoff}

\author{Ye Zhou (\begin{CJK*}{UTF8}{gbsn}周烨\end{CJK*})\orcidlink{0009-0004-7050-7736}}
\email{ye.zhou.horizon@gmail.com}
\affiliation{Independent Researcher, Kunshan, Jiangsu, China}

\begin{abstract}
Direct canonical quantization of Jackiw--Teitelboim gravity at finite cutoff leads to a critically singular P\"oschl--Teller problem whose endpoint admits inequivalent self-adjoint domains. We solve the full one-parameter family, obtain its exact spectral data, and identify the extension parameter with a renormalized short-distance matching datum in geodesic-length space. Self-adjointness, positivity, and unitary Brown--York evolution leave a stable family, whereas a source-free local completion without an endpoint contact term selects the Friedrichs realization. With the Brown--York clock and a local cap prescription fixed, the disk amplitude distinguishes the domains. In the large-boundary limit the bulk operator becomes universally Liouville, while a cap tied to the receding endpoint can retain domain dependence. Analytic continuation also maps radial resonances to poles of the Brown--York resolvent. These results isolate the short-distance datum that distinguishes the reduced canonical realizations and that a microscopic completion would have to determine.
\end{abstract}

\maketitle

\section{Introduction}\label{Sec:Intro}

Jackiw--Teitelboim (JT) gravity is a useful laboratory for quantum gravity because both its boundary dynamics and its bulk quantization can be treated analytically \cite{Teitelboim1983,Jackiw1985,Henneaux1985}. At an asymptotic boundary the low-energy dynamics is governed by the Schwarzian theory \cite{MaldacenaStanfordYang2016}, whose partition function and correlation functions admit exact treatments \cite{StanfordWitten2017,MertensTuriaciVerlinde2017}. Canonical and representation-theoretic formulations give a detailed description of the JT Hilbert space and spectrum \cite{KitaevSuh2019,Yang2019,IliesiuPufuVerlindeWang2019}, while the topological expansion is captured by a double-scaled matrix integral \cite{SaadShenkerStanford2019}. At finite cutoff, however, the relation between boundary data and the quantum Hamiltonian becomes more subtle. The problem has been approached through Wheeler--DeWitt and $T\bar T$-type constructions \cite{IliesiuKruthoffTuriaciVerlinde2020,GrossKruthoffRolphShaghoulian2020TTbar,GrossKruthoffRolphShaghoulian2020Hamiltonian}, self-avoiding boundary curves \cite{StanfordYang2020}, nonperturbative completions \cite{GriguoloPaneraiPapaliniSeminara2022}, and BF and generalized boundary-condition formulations \cite{EbertFerkoSunSun2022,GoelIliesiuKruthoffYang2021}. More recently, finite-cutoff JT gravity has also been derived from holographic RG flow, including a one-parameter family of boundary conditions \cite{SheikhJabbariTaghiloo2026}. A complementary program seeks a microscopic description in terms of random disks and self-overlapping curves, both in the continuum and on the lattice \cite{Ferrari2025RandomDisks,ChaudhuriFerrari2025OneLoop,Ferrari2026Lattice}. Finite-cutoff matrix-model and late-time observables provide another complementary perspective \cite{BhattacharyyaEtAl2026}.

The formulation most directly relevant here is the recent canonical reduction of Refs.~\cite{GriguoloPapaliniRussoSeminara2026,GriguoloEtAl2026Box}. Using the geodesic length between the two boundaries as a reduced coordinate, the quantum problem becomes a generalized P\"oschl--Teller system with a critical inverse-square singularity at vanishing length. Both independent local solutions are square integrable there, so the differential expression alone does not specify a unique self-adjoint operator. Ref.~\cite{GriguoloEtAl2026Box} adopted the Friedrichs realization and obtained exact wavefunctions, the disk amplitude, correlation functions, and the Schwarzian limit. Critical inverse-square endpoints and generalized P\"oschl--Teller systems have a well-developed spectral theory \cite{Case1950,DerezinskiRichard2017,CamaraDaSilvaPereiraLima2024}; the gravitational question is what fixes the endpoint domain in the finite-boundary theory.

We address this question within the open-channel canonical description. We solve the full mixed endpoint family and identify its parameter with a renormalized short-distance matching datum. With the Brown--York clock and a local cap prescription fixed, different domains lead to distinguishable disk amplitudes. The standard canonical consistency conditions restrict the allowed domains but do not select one uniquely; a source-free completion without an endpoint contact term instead singles out the Friedrichs realization. The endpoint data can also survive limits in which the local bulk dynamics has become universal: the translated Hamiltonian approaches Liouville quantum mechanics, while an endpoint-attached cap can retain domain dependence, and radial resonances are carried by the Brown--York spectral map into the analytically continued Brown--York resolvent.

Section~\ref{Sec:Canonical} develops the canonical spectral problem. Section~\ref{Sec:Disk} relates the endpoint condition to its cutoff realization and the disk amplitude. Section~\ref{Sec:Selection} discusses the physical selection of the domain, and Sec.~\ref{Sec:Dynamics} studies the large-boundary limit and Brown--York dynamics. Section~\ref{Sec:Conclusion} concludes.

\section{Canonical quantization at finite cutoff}\label{Sec:Canonical}

The canonical treatment of JT gravity goes back to Ref.~\cite{Henneaux1985}. Here we use the finite-boundary formulation developed in Refs.~\cite{GriguoloPapaliniRussoSeminara2026,GriguoloEtAl2026Box}, and in particular the direct reduced quantization of Ref.~\cite{GriguoloEtAl2026Box}. We set $8\pi G_{\rm N}=\ell_{\rm AdS}=\hbar=1$ and omit the topological Euler term. With the standard finite-Dirichlet-boundary counterterm, the Lorentzian action is
\begin{equation}
S_{\rm JT}=\frac12\int_{\mathcal M}\mathrm d^2x\,\sqrt{-g}\,\Phi(R+2)+\int_{\partial\mathcal M}\mathrm dy\,\sqrt{|\gamma|}\,\Phi(K-1).
\label{eq:JTaction}
\end{equation}
For the static exterior solution,
\begin{equation}
\Phi=r,\qquad \mathrm ds^2=-(r^2-k^2)\,\mathrm dt^2+\frac{\mathrm dr^2}{r^2-k^2},
\label{eq:staticJT}
\end{equation}
we denote the fixed boundary value of the dilaton by $\nu\equiv\Phi_b$ and its horizon value by $k\equiv\Phi_h$, with $0<k<\nu$. Writing $q=\sqrt{\nu^2-k^2}$, the Brown--York energy conjugate to the symmetric boundary clock is \cite{BrownYork1993,GriguoloEtAl2026Box}
\begin{equation}
E=\nu-q=\nu-\sqrt{\nu^2-k^2}.
\label{eq:BYclassical}
\end{equation}
The reduced phase space may be parametrized by the geodesic length $L$ joining the two boundaries and its conjugate momentum $P$. The reduction gives
\begin{equation}
k^2=P^2+\frac{\nu^2}{\cosh^2(L/2)},\qquad \mathrm dL\wedge\mathrm dP=\mathrm d\tau\wedge\mathrm dE.
\label{eq:reducedCasimir}
\end{equation}
Thus $k^2$ is the reduced Casimir, while $E$ is its nonlinear Brown--York energy coordinate. The exterior restriction $k\leq\nu$ implies $|P|\leq\nu\tanh(L/2)$, and the endpoint $L=0$ is reached at the phase-space edge $P=0$, $k=\nu$.

Quantizing Eq.~\eqref{eq:reducedCasimir} requires an ordering prescription. We adopt the representation-theoretic ordering of Ref.~\cite{GriguoloEtAl2026Box}, which follows from the radial reduction of the $\mathfrak{sl}(2,\mathbb R)$ Casimir and is closely related to the representation-theoretic description of JT gravity \cite{BlommaertMertensVerschelde2018,IliesiuPufuVerlindeWang2019} and to radial $\mathrm{SL}(2,\mathbb R)$ scattering problems \cite{InoueRichard2025}. In local coordinates $(a,L,b)$ the relevant radial expression is
\begin{equation}
\begin{aligned}
\Omega={}&4\partial_L^2+4\coth L\,\partial_L-\frac{\partial_a^2+\partial_b^2}{\sinh^2L}
+\frac{2\cosh L}{\sinh^2L}\,\partial_a\partial_b .
\end{aligned}
\label{eq:radialCasimir}
\end{equation}
Restricting to $\partial_a=\partial_b=2\mathrm i\nu$ and conjugating the radial measure to the flat measure by $U=\sqrt{\sinh L}$ gives the Schr\"odinger operator on $\mathcal H=L^2(\mathbb R_+,\mathrm dL)$,
\begin{equation}
\begin{aligned}
A_\nu&=-\frac14\left(U\Omega U^{-1}+1\right)\\
&=-\frac{\mathrm d^2}{\mathrm dL^2}+\frac{\nu^2}{\cosh^2(L/2)}-\frac{1}{4\sinh^2L}.
\end{aligned}
\label{eq:Anu}
\end{equation}
The last term changes the endpoint problem qualitatively. As $L\to0$,
\begin{equation}
\frac{\nu^2}{\cosh^2(L/2)}-\frac{1}{4\sinh^2L}
=-\frac{1}{4L^2}+\nu^2+\frac{1}{12}+O(L^2),
\label{eq:criticalPotential}
\end{equation}
so the indicial equation has the repeated root $r=\tfrac12$. The two independent local behaviors are therefore $\sqrt L$ and $\sqrt L\log L$. Both are square integrable at the origin, which is limit circle, whereas infinity is limit point. This is the critical inverse-square endpoint familiar from singular quantum mechanics \cite{Case1950,DerezinskiRichard2017,GitmanTyutinVoronov2010}; closely related generalized P\"oschl--Teller realizations with anomalous endpoint scales have been studied in Refs.~\cite{CamaraDaSilvaPereiraLima2024,DerezinskiKarimi2026}.

To fix conventions, let $s(L)=\sqrt{L/2}$ and choose the Frobenius pair
\begin{equation}
\begin{aligned}
\phi_z(L)&=s(L)\bigl[1+O(L^2)\bigr],\\
\theta_z(L)&=s(L)\bigl[-2\log(L/2)+O(L^2|\log L|)\bigr],\\
W(\theta_z,\phi_z)&=1,
\end{aligned}
\label{eq:FrobeniusBasis}
\end{equation}
where $W(u,v)=uv'-u'v$. Every function in the maximal Sturm--Liouville domain has an expansion
\begin{equation}
u(L)=s(L)\left[a(u)-2b(u)\log\frac{L}{2}\right]+o(\sqrt L),
\qquad L\to0,
\label{eq:endpointExpansion}
\end{equation}
with coefficients equivalently defined by
\begin{equation}
b(u)=\lim_{L\to0}W(u,\phi_0),\qquad
a(u)=-\lim_{L\to0}W(u,\theta_0).
\label{eq:boundaryCoefficients}
\end{equation}
For an inner product antilinear in its first argument, Green's identity reduces at the endpoint to
\begin{equation}
\langle u,A_\nu v\rangle-\langle A_\nu u,v\rangle
=\overline{b(u)}\,a(v)-\overline{a(u)}\,b(v).
\label{eq:boundaryForm}
\end{equation}
Standard singular Sturm--Liouville theory then gives the one-parameter family of self-adjoint realizations \cite{Teschl2014,BehrndtHassiDeSnoo2020},
\begin{equation}
\mathcal D(A_{\nu,h})
=\{u\in\mathcal D_{\max}:a(u)=h\,b(u)\},
\qquad h\in\mathbb R,
\label{eq:Dh}
\end{equation}
together with the Friedrichs realization
\begin{equation}
\mathcal D(A_{\nu,F})
=\{u\in\mathcal D_{\max}:b(u)=0\}.
\label{eq:DF}
\end{equation}
We shall use the boundary maps $\Gamma_0u=b(u)$ and $\Gamma_1u=a(u)$.

The exact continuum data can be obtained in the same normalization. Set
\begin{equation}
t=\tanh^2\frac{L}{2},\qquad
a_k=\frac12-\mathrm ik+\mathrm i\nu,\qquad
b_k=\frac12-\mathrm ik-\mathrm i\nu,\qquad
c_k=1-2\mathrm ik .
\label{eq:hypergeomPars}
\end{equation}
For $\operatorname{Im}k>0$, an outgoing solution of $(A_\nu-k^2)f_k=0$, normalized by $f_k(L)\sim e^{\mathrm ikL}$ as $L\to\infty$, is
\begin{equation}
f_k(L)=4^{\mathrm ik}t^{1/4}(1-t)^{-\mathrm ik}
\,{}_2F_1(a_k,b_k;c_k;1-t).
\label{eq:JostSolution}
\end{equation}
Since $c_k=a_k+b_k$, the zero-balanced hypergeometric connection formula \cite{NISTDLMF2026} gives
\begin{equation}
f_k(L)=Q_\nu(k)\left[\theta_{k^2}(L)+M_\nu(k)\phi_{k^2}(L)\right],
\label{eq:JostEndpoint}
\end{equation}
with
\begin{equation}
Q_\nu(k)=\frac{4^{\mathrm ik}\Gamma(1-2\mathrm ik)}
{\Gamma(a_k)\Gamma(b_k)},
\qquad
M_\nu(k)=-2\gamma_{\rm E}-\psi(a_k)-\psi(b_k),
\label{eq:Mnu}
\end{equation}
where $\psi=\Gamma'/\Gamma$. For $z=k^2$ with $\operatorname{Im}k>0$, define $m_\nu(z)=M_\nu(k)$ and $g_z=\theta_z+m_\nu(z)\phi_z$. Then $\Gamma_0g_z=1$, $\Gamma_1g_z=m_\nu(z)$, and Eq.~\eqref{eq:boundaryForm} gives
\begin{equation}
\frac{\operatorname{Im}m_\nu(z)}{\operatorname{Im}z}
=\|g_z\|^2>0.
\label{eq:Herglotz}
\end{equation}
Thus $m_\nu$ is the Weyl--Titchmarsh function in the present normalization. For real $k>0$, the digamma reflection identity reduces its imaginary part to
\begin{equation}
B_\nu(k)\equiv\operatorname{Im}M_\nu(k)
=\frac{\pi\sinh(2\pi k)}
{\cosh(2\pi k)+\cosh(2\pi\nu)}>0.
\label{eq:ImM}
\end{equation}
The corresponding continuum spectral measures are \cite{KostenkoSakhnovichTeschl2012}
\begin{equation}
\mathrm d\mu_F=\rho_F(k)\,\mathrm dk,\qquad
\rho_F(k)=\frac{2k}{\pi}B_\nu(k)
=\frac{2k\sinh(2\pi k)}
{\cosh(2\pi k)+\cosh(2\pi\nu)},
\label{eq:rhoF}
\end{equation}
and, for the boundary-normalized eigenfunction
$u_{h,k^2}=\theta_{k^2}+h\phi_{k^2}$,
\begin{equation}
\mathrm d\mu_h=\rho_h(k)\,\mathrm dk,\qquad
\rho_h(k)=\frac{\rho_F(k)}
{|h-M_\nu(k)|^2}.
\label{eq:rhoh}
\end{equation}
Equation~\eqref{eq:rhoF} reproduces the Friedrichs spectral density of Ref.~\cite{GriguoloEtAl2026Box}, while Eq.~\eqref{eq:rhoh} gives the continuum measure for the full mixed endpoint family.

The same Weyl function determines the lower edge of the spectrum. At zero energy,
\begin{equation}
h_c(\nu)=m_\nu(0)
=-2\gamma_{\rm E}
-2\operatorname{Re}\psi\!\left(\frac12+\mathrm i\nu\right).
\label{eq:hc}
\end{equation}
On the negative real axis, Eq.~\eqref{eq:Herglotz} implies
$m_\nu'(\lambda)=\|g_\lambda\|^2>0$, while the large-argument asymptotics of the digamma function gives
$m_\nu(-\kappa^2)\to-\infty$ as $\kappa\to\infty$ \cite{NISTDLMF2026}. Hence the equation $m_\nu(\lambda)=h$ has exactly one negative solution for $h<h_c(\nu)$ and none for $h\geq h_c(\nu)$. At $h=h_c(\nu)$ the zero-energy solution approaches a nonzero constant at infinity and is therefore a threshold resonance rather than an $L^2$ eigenstate. Since the essential spectrum is $[0,\infty)$ for every realization,
\begin{equation}
\begin{aligned}
h<h_c(\nu)&:\quad \text{one negative eigenvalue},\\
h=h_c(\nu)&:\quad \text{threshold resonance},\\
h>h_c(\nu)&:\quad A_{\nu,h}\geq0,\\
F&:\quad A_{\nu,F}\geq0.
\end{aligned}
\label{eq:stabilityClassification}
\end{equation}
We shall refer to the nonnegative realizations as the stable family.

For any stable realization, the exterior Brown--York Hamiltonian is obtained from the same spectral map as in Ref.~\cite{GriguoloEtAl2026Box}. On
\begin{equation}
\mathcal H_{\rm ext}
=\mathbf 1_{[0,\nu^2]}(A_{\nu,h})\mathcal H,
\end{equation}
it is
\begin{equation}
H^-_{\nu,h}
=\nu-\sqrt{\nu^2-A_{\nu,h}},
\qquad
0\leq H^-_{\nu,h}\leq\nu.
\label{eq:BYoperator}
\end{equation}
The canonical spectral problem is therefore completely fixed once an endpoint domain is chosen. What remains is to determine how that endpoint datum is realized at short distance and how it enters the state that closes the spatial interval. We turn to these questions in Sec.~\ref{Sec:Disk}.

\section{Endpoint conditions and the disk amplitude}\label{Sec:Disk}

Section~\ref{Sec:Canonical} fixes the interior differential operator but leaves its critical endpoint condition unspecified. Before asking whether gravity selects a particular domain, we first give the parameter $h$ a direct short-distance meaning and determine how it enters the state that closes the open channel.

To expose the short-distance meaning of $h$, undo the flat-measure conjugation in Eq.~\eqref{eq:Anu} and define
\begin{equation}
w(L)=\frac12\sinh L,\qquad u(L)=\sqrt{w(L)}\,v(L).
\label{eq:weightedTransform}
\end{equation}
Then
\begin{equation}
\begin{aligned}
T_\nu&\equiv w^{-1/2}A_\nu w^{1/2}
=-\frac1w\frac{\mathrm d}{\mathrm dL}
\left(w\frac{\mathrm d}{\mathrm dL}\right)+V_\nu(L),\\
V_\nu(L)&=\frac{\nu^2}{\cosh^2(L/2)}-\frac14.
\end{aligned}
\label{eq:Tnu}
\end{equation}
For the finite-$h$ domain $a=hb$, the endpoint expansion \eqref{eq:endpointExpansion} becomes
\begin{equation}
v(L)=b\left[h-2\log\frac{L}{2}\right]+o(1),
\qquad
w(L)v'(L)\longrightarrow-b.
\label{eq:vEndpoint}
\end{equation}
A regular cutoff at $L=\epsilon$ may therefore be equipped with the Robin condition
\begin{equation}
w(\epsilon)v'(\epsilon)=c_\epsilon v(\epsilon).
\label{eq:RobinCutoff}
\end{equation}
The endpoint expansion gives
\begin{equation}
c_\epsilon
=-\frac{1}{h-2\log(\epsilon/2)}
+o\!\left(\frac{1}{\log^2\epsilon}\right),
\label{eq:cEpsilon}
\end{equation}
or equivalently
\begin{equation}
h=\lim_{\epsilon\to0}
\left[
2\log\frac{\epsilon}{2}
-\frac{1}{c_\epsilon}
\right].
\label{eq:hMatching}
\end{equation}
This is the critical inverse-square renormalization mechanism in the present length variable \cite{CamblongEtAl2000,BouazizBawin2014,DerezinskiRichard2017}. Thus $h$ is a renormalized short-distance matching datum in geodesic-length configuration space. Its numerical value changes with the reference scale used inside the logarithm, while the domain itself does not.

The finite part in Eq.~\eqref{eq:hMatching} contains information that is lost if one keeps only $c_\epsilon\to0$. Writing
\begin{equation}
X_\epsilon=h-2\log(\epsilon/2),
\qquad
c_\epsilon=-X_\epsilon^{-1}+r_\epsilon,
\end{equation}
one finds
\begin{equation}
h_{\rm eff}(\epsilon)-h
=
\frac{X_\epsilon^2r_\epsilon}
{1-X_\epsilon r_\epsilon}.
\label{eq:matchingSensitivity}
\end{equation}
Corrections of order $X_\epsilon^{-2}$ can therefore shift the limiting value of $h$ by a finite amount.

Every stable finite $h$ also admits a nonnegative cutoff realization. Let $u_{h,0}$ denote the zero-energy solution with endpoint data $a=h$, $b=1$, and set $p_h=u_{h,0}/\sqrt w$. For $h\geq h_c(\nu)$,
\begin{equation}
u_{h,0}=g_0+(h-h_c)\phi_0
\end{equation}
is positive on $L>0$, as follows from the zero-energy hypergeometric representatives. Define
\begin{equation}
c_\epsilon
=
w(\epsilon)\frac{p_h'(\epsilon)}{p_h(\epsilon)}.
\label{eq:exactCutoffCoefficient}
\end{equation}
The cutoff quadratic form is then
\begin{equation}
\begin{aligned}
\mathfrak q_{\epsilon,h}[v]
={}&
\int_\epsilon^\infty\!\mathrm dL\,
w\left(|v'|^2+V_\nu|v|^2\right)\\
&+c_\epsilon|v(\epsilon)|^2,
\end{aligned}
\label{eq:cutoffForm}
\end{equation}
and the zero-energy equation gives the ground-state representation
\begin{equation}
\mathfrak q_{\epsilon,h}[v]
=
\int_\epsilon^\infty\!\mathrm dL\,
w p_h^2
\left|
\left(\frac{v}{p_h}\right)'
\right|^2
\geq0.
\label{eq:groundStateRepresentation}
\end{equation}
If the left solution at spectral parameter $z$ is written as $\theta_z+h_\epsilon(z)\phi_z$, the Robin condition gives
\begin{equation}
\begin{aligned}
h_\epsilon(z)-h
={}&
z\int_0^\epsilon\!\mathrm dL\,
u_{h,0}(L)\\
&\times
\left[
\theta_z(L)+h_\epsilon(z)\phi_z(L)
\right]
=
O\!\left(\epsilon^2|\log\epsilon|^2\right).
\end{aligned}
\label{eq:hEpsilonConvergence}
\end{equation}
The Green kernels therefore converge locally. This gives weak resolvent convergence on a dense set; the uniform self-adjoint resolvent bound and the standard resolvent norm identity then yield strong-resolvent convergence to $A_{\nu,h}$ \cite{Kato1995}. Thus every stable finite $h$ arises as the limit of regular cutoff problems obeying the matching law \eqref{eq:hMatching}.

We next determine the endpoint state that closes the spatial interval. Ordinary evaluation at $L=0$ is unsuitable because both local solutions vanish there. For finite $h$, the boundary-normalized continuum solution
\begin{equation}
u_{h,k^2}(L)
=
\theta_{k^2}(L)+h\phi_{k^2}(L)
\label{eq:uhk}
\end{equation}
satisfies
\begin{equation}
\frac{
u_{h,k^2}(L)
}{
\sqrt{L/2}\,[h-2\log(L/2)]
}
\longrightarrow1,
\qquad
L\to0,
\label{eq:shrinkingCapFiniteH}
\end{equation}
uniformly for $k$ in compact intervals. The Friedrichs solution similarly satisfies
\begin{equation}
\frac{\phi_{k^2}(L)}{\sqrt{L/2}}
\longrightarrow1.
\end{equation}
With the boundary maps of Sec.~\ref{Sec:Canonical}, we therefore define
\begin{equation}
\mathcal B_hu=\Gamma_0u=b(u),
\qquad
\mathcal B_Fu=\Gamma_1u=a(u).
\label{eq:capFunctional}
\end{equation}
On the finite-$h$ domain, any coefficient-local linear trace $\alpha\Gamma_0+\beta\Gamma_1$ reduces to $(\alpha+\beta h)\Gamma_0$; for the Friedrichs domain the corresponding nontrivial trace is $\Gamma_1$. The endpoint functional is therefore fixed up to an energy-independent normalization, which we denote by $C_h$. In the spectral representation its wavefunction is constant. In the Friedrichs sector this reduces to the renormalized vanishing-length prescription of Ref.~\cite{GriguoloEtAl2026Box}.

On the exterior band $0\leq k\leq\nu$, the Brown--York energy is
\begin{equation}
E_\nu(k)
=
\nu-\sqrt{\nu^2-k^2},
\label{eq:EkDisk}
\end{equation}
and the corresponding cap amplitude is
\begin{equation}
Z_h(\beta)
=
|C_h|^2
\int_0^\nu\!\mathrm dk\,
\rho_h(k)e^{-\beta E_\nu(k)}.
\label{eq:ZhDisk}
\end{equation}
Choosing $C_F$ in the normalization of Ref.~\cite{GriguoloEtAl2026Box}, Eq.~\eqref{eq:ZhDisk} reduces to its finite-boundary disk partition function. For finite $h$, the same construction defines the canonical amplitude associated with the local endpoint trace \eqref{eq:capFunctional}.

Positivity and interval composition are immediate in the spectral representation. Define
\begin{equation}
\Psi_{h,t}(k)
=
C_h e^{-tE_\nu(k)}.
\label{eq:heatCapWavefunction}
\end{equation}
Then, in $L^2([0,\nu],\mathrm d\mu_h)$,
\begin{equation}
Z_h(\beta)
=
\|\Psi_{h,\beta/2}\|^2,
\qquad
\langle\Psi_{h,t_1},\Psi_{h,t_2}\rangle
=
Z_h(t_1+t_2).
\label{eq:capGluing}
\end{equation}

Within this endpoint family, exact disk data distinguish the domain once the clock and the coefficient-local cap are fixed independently. Suppose two finite parameters $h_1$ and $h_2$, with nonzero energy-independent calibrations, give the same $Z(\beta)$ for all $\beta>0$. Uniqueness of the Laplace transform \cite{SchillingSongVondracek2012} implies that their finite-band energy measures are proportional. Since $E_\nu(k)$ is strictly increasing for $0<k<\nu$,
\begin{equation}
\rho_{h_1}(k)
=
c\,\rho_{h_2}(k),
\qquad
0<k<\nu,
\label{eq:proportionalMeasures}
\end{equation}
for some $c>0$. Within the family \eqref{eq:rhoh}, both sides are real analytic for $k>0$, so the relation extends to the positive axis. Their ratio tends to one as $k\to\infty$, hence $c=1$. Equality of the two denominators then gives
\begin{equation}
(h_1-h_2)
\left[
h_1+h_2-2\operatorname{Re}M_\nu(k)
\right]
=0.
\label{eq:diskIdentifiability}
\end{equation}
Because $\operatorname{Re}M_\nu(k)$ is not constant, $h_1=h_2$. A finite-$h$ measure cannot be proportional to the Friedrichs measure either, since
\begin{equation}
\frac{\rho_h(k)}{\rho_F(k)}
\longrightarrow0,
\qquad
k\to\infty.
\end{equation}

The assumption that the cap is fixed independently is essential. If an arbitrary spectral cap is allowed, any positive target density can formally be reproduced by choosing
\begin{equation}
c_h(k)
=
\sqrt{
\frac{\rho_{\rm target}(k)}
{\rho_h(k)}
},
\label{eq:nonlocalCompensatingCap}
\end{equation}
whenever the resulting state is admissible. Such a weight is generally nonlocal in the endpoint data. Thus a disk amplitude by itself does not select a domain; the domain becomes identifiable after the clock and local cap prescription have been fixed.

The endpoint parameter now has a precise short-distance realization, and a chosen domain determines its canonical cap amplitude. The remaining question is whether the gravitational theory supplies an additional principle that selects one member of the stable family. We address this in Sec.~\ref{Sec:Selection}.

\section{Physical selection of the endpoint condition}\label{Sec:Selection}

The endpoint family is now explicit. The remaining question is whether the gravitational construction removes this freedom. Within the reduced canonical data fixed in Secs.~\ref{Sec:Canonical} and \ref{Sec:Disk}, it does not. For every $h\geq h_c(\nu)$, as well as for the Friedrichs realization, $A_{\nu,h}$ is self-adjoint and nonnegative, the continuum measure is positive, the exterior Brown--York Hamiltonian \eqref{eq:BYoperator} generates unitary evolution, and the heat-regularized cap amplitudes \eqref{eq:ZhDisk} are positive and obey the interval composition law \eqref{eq:capGluing}. These conditions restrict the theory to the stable family but do not choose a member of it.

Nor can two values of $h$ be identified by a change of basis if the length observable is kept fixed. Suppose a unitary $U$ intertwines $A_{\nu,h_1}$ and $A_{\nu,h_2}$ and commutes with every bounded multiplication operator $F(L)$. Then $U$ itself is multiplication by a phase, $Uu=\chi(L)u$. Intertwining on $C_c^\infty(0,\infty)$ implies $\chi\in H^2_{\rm loc}$ and
\begin{equation}
-2\chi'(L)u'(L)-\chi''(L)u(L)=0
\end{equation}
for every test function $u$. Hence $\chi'=0$, so $U$ is a constant phase. Equality of the two operator domains then requires $h_1=h_2$. Distinct endpoint parameters therefore define inequivalent pairs $(A_{\nu,h},L)$.

There is, however, a natural additional condition that does select the Friedrichs domain. In the weighted variable of Sec.~\ref{Sec:Disk}, a finite-$h$ stationary state has
\begin{equation}
v(L)=a-2b\log\frac{L}{2}+o(1),\qquad
w(L)v'(L)\longrightarrow-b.
\label{eq:selectionEndpoint}
\end{equation}
Since $w(L)=L/2+O(L^3)$,
\begin{equation}
\int_\epsilon^{L_0}\!\mathrm dL\,w(L)|v'(L)|^2
=
2|b|^2\log\frac1\epsilon+O(1).
\label{eq:noContactDivergence}
\end{equation}
Thus the logarithmic branch is absent if the radial Hamiltonian is defined by closing the unrenormalized local quadratic form without an endpoint contact counterterm. The resulting condition is $b=0$, precisely the Friedrichs domain.

The same conclusion appears in the weak equation. For $T_\nu v=\lambda v$ away from the origin and a smooth test function $\chi$ supported up to $L=0$, integration by parts gives
\begin{equation}
\int_0^\infty\!\mathrm dL\,
\left(
wv'\overline{\chi'}
+wV_\nu v\bar\chi
-\lambda wv\bar\chi
\right)
=
b\,\overline{\chi(0)}.
\label{eq:weakEndpointSource}
\end{equation}
The coefficient $b$ is therefore the strength of a distributional endpoint source in the reduced radial equation. Requiring a source-free completion again gives $b=0$. This condition is stronger than vanishing probability current: for every self-adjoint boundary relation $a=hb$ with real $h$,
\begin{equation}
2\operatorname{Im}(\bar v\,wv')
\longrightarrow
-2\operatorname{Im}(\bar a b)
=
0.
\end{equation}
Friedrichs is therefore selected by the additional no-contact condition, not by self-adjointness or unitary evolution alone.

It remains to ask whether this extra condition follows from gravity itself. In the metric formulation, intrinsic local counterterms on a one-dimensional Dirichlet wall with fixed $\Phi_b=\nu$ reduce to the boundary length multiplied by a function of $\nu$. They can shift the boundary energy or the overall normalization, but they do not change the logarithmic relation between $a$ and $b$ for the fixed operator \eqref{eq:Anu}. Changing the boundary ensemble, adding extrinsic terms, corners or localized defects constitutes different variational data \cite{GoelIliesiuKruthoffYang2021,MoitraSakeTrivedi2021}. A smooth classical cap is likewise not yet a quantum statement about the form domain. If the regulated metric path integral induces the unrenormalized source-free form used above, it selects Friedrichs; establishing that implication requires the off-shell measure, the corner prescription for a shrinking spatial cut and the corresponding map to the canonical length wavefunction. Existing finite-boundary metric constructions do not directly provide this marked-cut-to-domain map \cite{GriguoloPapaliniRussoSeminara2026,ChaudhuriFerrari2025OneLoop}.

The BF formulation does not fix the endpoint merely by replacing the radial problem with a group-theoretic one. Exact JT quantization depends on global group data, boundary restrictions and the physical inner product in addition to the local Casimir \cite{IliesiuPufuVerlindeWang2019,BlommaertMertensVerschelde2018,BlommaertMertensVerschelde2019,Ferrari2021Gauge}. In particular, smoothness before reduction need not imply the absence of the logarithmic branch after reduction. Near the identity, use smooth coordinates $(c,Q,T)$ adapted to a noncompact orbit, with
\begin{equation}
Q=L\cosh\eta,\qquad
T=L\sinh\eta.
\end{equation}
Choose a smooth compactly supported group function whose transverse dependence is $\chi(Q^2+T^2)$, with $\chi=1$ near the origin. For equal left and right hyperbolic charges, the relative boost $\eta$ is unweighted and the reduced radial factor is
\begin{equation}
\begin{aligned}
I_\chi(L)
&=
\int_{\mathbb R}\!\mathrm d\eta\,
\chi\!\left(L^2\cosh2\eta\right)\\
&=
\int_{L^2}^{\infty}
\frac{\chi(s)\,\mathrm ds}{\sqrt{s^2-L^4}}
=
-2\log L+C_\chi+O(L^4).
\end{aligned}
\label{eq:BFlogCounterexample}
\end{equation}
A smooth group function can therefore reduce to a logarithmic radial function. Smoothness of the unreduced wavefunction is consequently not, by itself, a criterion for $b=0$; the global BF domain must be determined together with the full reduction and boundary data.

A microscopic random-disk construction offers another possible source of the missing short-distance datum. Ferrari's continuum and lattice formulations define JT gravity in terms of constant-curvature fillings and self-overlapping boundaries, with an independent one-loop analysis of the finite-cutoff disk \cite{Ferrari2025RandomDisks,ChaudhuriFerrari2025OneLoop,Ferrari2026Lattice}. The lattice construction already includes open self-overlapping curves \cite{Ferrari2026Lattice}. To compare it with the canonical endpoint problem, however, one needs a marked two-cut transfer amplitude that retains the intrinsic distance between the marked points, an additive evolution parameter, and a map to the canonical variables $(L,H^-_{\nu,h})$. Its short-distance logarithmic derivative could then be matched to $c_\epsilon$, and Eq.~\eqref{eq:hMatching} would determine $h$. The marginal disk partition function has already integrated over the marked-cut information needed for this step. The required comparison is therefore a transfer-kernel dictionary, rather than a matching of disk marginals alone.

Within the reduced canonical theory, the result is consequently simple: the usual consistency conditions leave an inequivalent stable family, while the source-free, no-contact completion selects Friedrichs. Whether a microscopic realization of finite-boundary JT gravity supplies that completion or a nontrivial finite matching datum is additional information. We now turn to consequences that can be derived for a chosen domain without assuming either outcome.

\section{Large-boundary limit and Brown--York dynamics}\label{Sec:Dynamics}

The large-boundary limit separates two effects of the endpoint data. Local bulk dynamics become universal, while states and analytic continuations tied to the receding endpoint can retain information about the chosen domain. Set
\begin{equation}
\ell_\nu=2\log(2\nu),\qquad x=L-\ell_\nu .
\label{eq:largeBoundaryShift}
\end{equation}
On every compact interval in $x$,
\begin{equation}
\frac{\nu^2}{\cosh^2[(x+\ell_\nu)/2]}-\frac{1}{4\sinh^2(x+\ell_\nu)}
=e^{-x}+O(\nu^{-2}),
\label{eq:LiouvillePotentialLimit}
\end{equation}
with a uniform error. The endpoint simultaneously moves to $x=-\ell_\nu\to-\infty$. Let $J_\nu:L^2(\mathbb R_+,\mathrm dL)\to L^2(\mathbb R,\mathrm dx)$ be the translation isometry
\begin{equation}
(J_\nu f)(x)=\mathbf 1_{(-\ell_\nu,\infty)}(x)f(x+\ell_\nu).
\label{eq:translationIsometry}
\end{equation}
Every compactly supported test function eventually lies away from the moving endpoint, and the translated differential expressions converge there to
\begin{equation}
T_\infty=-\frac{\mathrm d^2}{\mathrm dx^2}+e^{-x}.
\label{eq:LiouvilleOperator}
\end{equation}
Since $T_\infty$ is essentially self-adjoint on $C_c^\infty(\mathbb R)$, standard core and resolvent arguments give, for $z\in\mathbb C\setminus\mathbb R$,
\begin{equation}
J_\nu(A_{\nu,h_\nu}-z)^{-1}J_\nu^*
\longrightarrow
(T_\infty-z)^{-1}
\label{eq:translatedResolventLimit}
\end{equation}
strongly on $L^2(\mathbb R)$, independently of the trajectory $h=h_\nu$ \cite{Kato1995,Teschl2014}. This is the Liouville quantum-mechanics operator underlying the Schwarzian limit of JT gravity \cite{MaldacenaStanfordYang2016,StanfordWitten2017,MertensTuriaciVerlinde2017,GriguoloEtAl2026Box}.

For a stable trajectory, let
\begin{equation}
P_{\nu,h_\nu}=\mathbf 1_{[0,\nu^2]}(A_{\nu,h_\nu})
\end{equation}
be the exterior spectral projection. The Brown--York spectral map satisfies
\begin{equation}
\nu\left(\nu-\sqrt{\nu^2-\lambda}\right)
=\frac{\lambda}{1+\sqrt{1-\lambda/\nu^2}}
\longrightarrow\frac{\lambda}{2},
\label{eq:BYScaling}
\end{equation}
and therefore, after the same translation,
\begin{equation}
J_\nu P_{\nu,h_\nu}
e^{-\nu tH^-_{\nu,h_\nu}}
P_{\nu,h_\nu}J_\nu^*
\longrightarrow
e^{-tT_\infty/2},
\qquad t>0,
\label{eq:fixedVectorHeatLimit}
\end{equation}
strongly on fixed vectors.

The cap behaves differently because it remains attached to $L=0$ and hence follows the moving point $x=-\ell_\nu$. For fixed $z$ away from the positive spectral cut, the Weyl function has the large-$\nu$ expansion
\begin{align}
h_c(\nu)&=-2\gamma_{\rm E}-2\log\nu+\frac{1}{12\nu^2}+O(\nu^{-4}),
\label{eq:hcLargeNu}\\
m_\nu(z)-h_c(\nu)&=\frac{z}{\nu^2}+\frac{z^2/2+z/4}{\nu^4}+O(\nu^{-6}),
\label{eq:mLargeNu}
\end{align}
from the standard digamma asymptotics \cite{NISTDLMF2026}. Let $\mathrm d\widehat\mu_h(\lambda)$ denote the spectral measure in the variable $\lambda=k^2$. Its Stieltjes transform is
\begin{equation}
q_h(z)=\frac{1}{h-m_\nu(z)}
=\int_0^\infty\frac{\mathrm d\widehat\mu_h(\lambda)}{\lambda-z}.
\label{eq:qhz}
\end{equation}
For the stable double scaling
\begin{equation}
h_\nu=h_c(\nu)+\frac{c}{\nu^2},\qquad c\geq0,
\label{eq:atomicScaling}
\end{equation}
Eqs.~\eqref{eq:mLargeNu} and \eqref{eq:qhz} give
\begin{equation}
\nu^{-2}q_{h_\nu}(z)\longrightarrow\frac{1}{c-z}.
\label{eq:StieltjesAtom}
\end{equation}
The right-hand side is the Stieltjes transform of a unit atom at $\lambda=c$. To state the measure limit without assuming finite total mass, define
\begin{equation}
\mathrm d\omega_\nu(\lambda)
=\frac{\mathrm d\widehat\mu_{h_\nu}(\lambda)}
{\nu^2(1+\lambda)}.
\label{eq:weightedMeasure}
\end{equation}
Its total mass is $\nu^{-2}q_{h_\nu}(-1)\to(1+c)^{-1}$. Partial fractions applied to Eq.~\eqref{eq:qhz} show that the Stieltjes transforms of $\omega_\nu$ converge to that of $\delta_c/(1+c)$; the convergence of the total masses then excludes loss of weighted mass at infinity. Hence
\begin{equation}
\omega_\nu\Longrightarrow\frac{\delta_c}{1+c}.
\end{equation}
Combining this result with $\nu E_\nu(\sqrt\lambda)\to\lambda/2$ on compact intervals and
\begin{equation}
\nu E_\nu(\sqrt\lambda)\geq\frac{\lambda}{2}
\end{equation}
on the exterior band gives, for the unit coefficient-local cap normalization,
\begin{equation}
\nu^{-2}Z_{h_\nu}(\nu t)\longrightarrow e^{-tc/2},
\qquad t>0.
\label{eq:movingCapAtom}
\end{equation}
Thus the same universal Liouville bulk operator can coexist with inequivalent moving-cap limits. Strong-resolvent convergence of the bulk Hamiltonian does not by itself determine the state attached to the receding endpoint.

We next turn to resonances, defined as poles of meromorphically continued localized resolvent matrix elements \cite{DyatlovZworski2019,FabrisRicharteSaa2021}. Related generalized P\"oschl--Teller problems with singular endpoint data were analyzed in Ref.~\cite{CamaraDaSilvaPereiraLima2024}. Let
\begin{equation}
P_{\nu,h}=\mathbf 1_{[0,\nu^2]}(A_{\nu,h}),\qquad
H=H^-_{\nu,h},\qquad
\widetilde H=2\nu-H,
\end{equation}
and define
\begin{equation}
g(\omega)=2\nu\omega-\omega^2.
\label{eq:gOmega}
\end{equation}
On $P_{\nu,h}\mathcal H$,
\begin{equation}
(H-\omega)(\widetilde H-\omega)
=A_{\nu,h}-g(\omega),
\end{equation}
so spectral calculus gives
\begin{equation}
R_H(\omega)+R_{\widetilde H}(\omega)
=g'(\omega)P_{\nu,h}R_A(g(\omega))P_{\nu,h},
\label{eq:BYResolventIdentity}
\end{equation}
where $R_A(z)=(A_{\nu,h}-z)^{-1}$. Suppose the outgoing continuation of $R_A$ has a simple pole at $z_r=k_r^2$, with $\operatorname{Re}k_r>0$ and $\operatorname{Im}k_r<0$. Away from the threshold and isolated anomalous Jost points, the companion Brown--York branch and the complementary auxiliary band are holomorphic near the corresponding point. Equation~\eqref{eq:BYResolventIdentity} therefore transports the pole to
\begin{equation}
\omega_r=\nu-\sqrt{\nu^2-k_r^2},
\qquad
\operatorname{Re}\sqrt{\nu^2-k_r^2}>0,
\label{eq:BYPoleMap}
\end{equation}
with
\begin{equation}
\operatorname*{Res}_{\omega=\omega_r}R_H(\omega)
=\operatorname*{Res}_{z=z_r}R_A(z).
\label{eq:BYResidueMap}
\end{equation}
The factor $g'(\omega_r)$ cancels the Jacobian from $z=g(\omega)$, so the Brown--York pole is a pole of the analytically continued canonical resolvent rather than a frequency obtained by a formal substitution.

A long-lived family appears at fixed positive distance from the stability threshold. Write
\begin{equation}
\begin{aligned}
h_\nu&=h_c(\nu)+\delta,\qquad \delta>0,\\
\alpha&=\sqrt{1-e^{-\delta}},\qquad
c_\alpha=\sqrt{1-\alpha^2}.
\end{aligned}
\label{eq:fixedOffsetScaling}
\end{equation}
In the region relevant to this root the Jost prefactor is regular, so the outgoing pole condition reduces to $h_\nu=M_\nu(k)$. Uniformly for $\alpha$ in compact subsets of $(0,1)$, the resulting simple root obeys
\begin{equation}
\begin{aligned}
\operatorname{Re}k_r&=\alpha\nu+O(\nu^{-1}),\\
-\operatorname{Im}k_r&=
\frac{\pi\nu c_\alpha^2}{2\alpha}
e^{-2\pi\nu(1-\alpha)}
\left[1+O(\nu^{-1})\right].
\end{aligned}
\label{eq:kLongLived}
\end{equation}
Mapping the pole through Eq.~\eqref{eq:BYPoleMap} gives
\begin{equation}
\operatorname{Re}\omega_r
=\nu(1-c_\alpha)+O(\nu^{-1}),
\end{equation}
and the Brown--York decay width
\begin{equation}
\Gamma\equiv-2\operatorname{Im}\omega_r
=\pi\nu c_\alpha
e^{-2\pi\nu(1-\alpha)}
\left[1+O(\nu^{-1})\right].
\label{eq:BYWidth}
\end{equation}

The exponential suppression has an independent semiclassical check. The critical $-1/(4L^2)$ term produces an inner boundary layer of width $L=O(\nu^{-1})$, which controls the endpoint matching but changes the barrier action only at subleading order. Outside this layer, the leading outer turning point is
\begin{equation}
L_t=2\operatorname{arcosh}(1/\alpha),
\end{equation}
and the leading WKB action is
\begin{equation}
\begin{aligned}
S&=\nu\int_0^{L_t}\!\mathrm dL\,
\sqrt{\operatorname{sech}^2(L/2)-\alpha^2}\\
&=\pi\nu(1-\alpha).
\end{aligned}
\label{eq:barrierAction}
\end{equation}
Hence
\begin{equation}
e^{-2S}=e^{-2\pi\nu(1-\alpha)},
\end{equation}
in agreement with Eq.~\eqref{eq:BYWidth}. The critical endpoint controls the inner matching and the prefactor, while the barrier action fixes the nonperturbative exponent.

The two effects have a common origin: sending the endpoint to infinity removes the domain from the local bulk limit without necessarily removing it from quantities tied to that endpoint. The finite-$h$ Brown--York resonances are consequently predictions of the chosen canonical domain. Their realization in a microscopic finite-boundary theory depends on the domain-selection problem discussed in Sec.~\ref{Sec:Selection}.

\section{Conclusion and discussion}\label{Sec:Conclusion}

The direct canonical formulation of finite-boundary JT gravity fixes the bulk differential operator but leaves a genuine endpoint datum. We have solved the corresponding one-parameter family, identified its parameter with a renormalized short-distance matching condition, and shown that the usual canonical consistency requirements leave a stable set of admissible domains. A stronger local condition does select a unique realization: removing endpoint contact terms and sources eliminates the logarithmic branch and gives the Friedrichs domain.

This distinction remains visible beyond the static spectrum. The large-boundary bulk Hamiltonian has a universal Liouville limit, but a state tied to the receding endpoint need not lose its domain dependence. Analytic continuation gives a second probe, since radial resonances are transported by the Brown--York spectral map to poles of the Brown--York resolvent. The endpoint condition therefore carries information that is not exhausted by the universal asymptotic bulk dynamics.

The remaining problem is to derive this datum from a microscopic gravitational construction. Metric, BF, and random-disk formulations provide complementary descriptions of finite-boundary JT gravity \cite{GriguoloPapaliniRussoSeminara2026,Ferrari2025RandomDisks,Ferrari2026Lattice}, but their short-distance data still have to be related directly to the canonical geodesic-length domain. In the random-disk approach, a natural target is a marked transfer amplitude that retains the intrinsic separation and its evolution before the disk data are marginalized. Establishing such a map would determine whether the microscopic theory realizes the no-contact Friedrichs completion or a different renormalized endpoint condition, and would provide the appropriate starting point for correlation functions and multi-boundary amplitudes.
\section*{Statements and Declarations}

\paragraph{Funding}
No funds, grants, or other support were received.

\paragraph{Competing interests}
The author has no relevant financial or non-financial interests to disclose.

\paragraph{Author contributions}
Ye Zhou: Conceptualization, Methodology, Formal analysis, Investigation, Writing -- original draft, Writing -- review and editing.

\paragraph{Data availability}
No datasets were generated or analyzed during this study.

\paragraph{Use of generative AI}
The author used OpenAI ChatGPT (GPT-5.6 Sol, accessed September 2026) for literature organization, language editing, and auxiliary checks of selected calculations. All derivations, mathematical arguments, scientific interpretations, and literature claims were independently verified by the author, who takes full responsibility for the manuscript.

\bibliography{reference}

@article{Teitelboim1983,
  author = {Teitelboim, Claudio},
  title = {Gravitation and hamiltonian structure in two spacetime dimensions},
  journal = {Phys. Lett. B},
  volume = {126},
  number = {1--2},
  pages = {41--45},
  year = {1983},
  doi = {10.1016/0370-2693(83)90012-6},
}

@article{Jackiw1985,
  author = {Jackiw, Roman},
  title = {Lower dimensional gravity},
  journal = {Nucl. Phys. B},
  volume = {252},
  pages = {343--356},
  year = {1985},
  doi = {10.1016/0550-3213(85)90448-1},
}

@article{Henneaux1985,
  author = {Henneaux, Marc},
  title = {Quantum Gravity in Two Dimensions: Exact Solution of the {Jackiw} Model},
  journal = {Phys. Rev. Lett.},
  volume = {54},
  number = {10},
  pages = {959--962},
  year = {1985},
  doi = {10.1103/PhysRevLett.54.959},
}

@article{BrownYork1993,
  author = {Brown, J. David and York, Jr., James W.},
  title = {Quasilocal energy and conserved charges derived from the gravitational action},
  journal = {Phys. Rev. D},
  volume = {47},
  number = {4},
  pages = {1407--1419},
  year = {1993},
  doi = {10.1103/PhysRevD.47.1407},
  eprint = {gr-qc/9209012},
  archivePrefix = {arXiv},
}

@article{MaldacenaStanfordYang2016,
  author = {Maldacena, Juan and Stanford, Douglas and Yang, Zhenbin},
  title = {Conformal symmetry and its breaking in two dimensional nearly anti-de-{Sitter} space},
  journal = {Prog. Theor. Exp. Phys.},
  volume = {2016},
  number = {12},
  pages = {12C104},
  year = {2016},
  doi = {10.1093/ptep/ptw124},
  eprint = {1606.01857},
  archivePrefix = {arXiv},
}

@article{StanfordWitten2017,
  author = {Stanford, Douglas and Witten, Edward},
  title = {Fermionic localization of the {Schwarzian} theory},
  journal = {JHEP},
  volume = {2017},
  number = {10},
  pages = {008},
  year = {2017},
  doi = {10.1007/JHEP10(2017)008},
  eprint = {1703.04612},
  archivePrefix = {arXiv},
}

@article{MertensTuriaciVerlinde2017,
  author = {Mertens, Thomas G. and Turiaci, Gustavo J. and Verlinde, Herman L.},
  title = {Solving the {Schwarzian} via the conformal bootstrap},
  journal = {JHEP},
  volume = {2017},
  number = {08},
  pages = {136},
  year = {2017},
  doi = {10.1007/JHEP08(2017)136},
  eprint = {1705.08408},
  archivePrefix = {arXiv},
}

@article{KitaevSuh2019,
  author = {Kitaev, Alexei and Suh, S. Josephine},
  title = {Statistical mechanics of a two-dimensional black hole},
  journal = {JHEP},
  volume = {2019},
  number = {05},
  pages = {198},
  year = {2019},
  doi = {10.1007/JHEP05(2019)198},
  eprint = {1808.07032},
  archivePrefix = {arXiv},
}

@article{Yang2019,
  author = {Yang, Zhenbin},
  title = {The quantum gravity dynamics of near extremal black holes},
  journal = {JHEP},
  volume = {2019},
  number = {05},
  pages = {205},
  year = {2019},
  doi = {10.1007/JHEP05(2019)205},
  eprint = {1809.08647},
  archivePrefix = {arXiv},
}

@article{IliesiuPufuVerlindeWang2019,
  author = {Iliesiu, Luca V. and Pufu, Silviu S. and Verlinde, Herman and Wang, Yifan},
  title = {An exact quantization of {Jackiw--Teitelboim} gravity},
  journal = {JHEP},
  volume = {2019},
  number = {11},
  pages = {091},
  year = {2019},
  doi = {10.1007/JHEP11(2019)091},
  eprint = {1905.02726},
  archivePrefix = {arXiv},
}

@article{BlommaertMertensVerschelde2019,
  author = {Blommaert, Andreas and Mertens, Thomas G. and Verschelde, Henri},
  title = {Fine structure of {Jackiw--Teitelboim} quantum gravity},
  journal = {JHEP},
  volume = {2019},
  number = {09},
  pages = {066},
  year = {2019},
  doi = {10.1007/JHEP09(2019)066},
  eprint = {1812.00918},
  archivePrefix = {arXiv},
}

@article{BlommaertMertensVerschelde2018,
  author = {Blommaert, Andreas and Mertens, Thomas G. and Verschelde, Henri},
  title = {The {Schwarzian} theory --- a {Wilson} line perspective},
  journal = {JHEP},
  volume = {2018},
  number = {12},
  pages = {022},
  year = {2018},
  doi = {10.1007/JHEP12(2018)022},
  eprint = {1806.07765},
  archivePrefix = {arXiv},
}

@misc{SaadShenkerStanford2019,
  author = {Saad, Phil and Shenker, Stephen H. and Stanford, Douglas},
  title = {{JT} gravity as a matrix integral},
  year = {2019},
  eprint = {1903.11115},
  archivePrefix = {arXiv},
  primaryClass = {hep-th},
}

@article{GrossKruthoffRolphShaghoulian2020TTbar,
  author = {Gross, David J. and Kruthoff, Jorrit and Rolph, Andrew and Shaghoulian, Edgar},
  title = {{$T\overline{T}$} in {$AdS_2$} and quantum mechanics},
  journal = {Phys. Rev. D},
  volume = {101},
  number = {2},
  pages = {026011},
  year = {2020},
  doi = {10.1103/PhysRevD.101.026011},
  eprint = {1907.04873},
  archivePrefix = {arXiv},
}

@article{GrossKruthoffRolphShaghoulian2020Hamiltonian,
  author = {Gross, David J. and Kruthoff, Jorrit and Rolph, Andrew and Shaghoulian, Edgar},
  title = {Hamiltonian deformations in quantum mechanics, {$T\overline{T}$}, and the {SYK} model},
  journal = {Phys. Rev. D},
  volume = {102},
  number = {4},
  pages = {046019},
  year = {2020},
  doi = {10.1103/PhysRevD.102.046019},
  eprint = {1912.06132},
  archivePrefix = {arXiv},
}

@article{IliesiuKruthoffTuriaciVerlinde2020,
  author = {Iliesiu, Luca V. and Kruthoff, Jorrit and Turiaci, Gustavo J. and Verlinde, Herman},
  title = {{JT} gravity at finite cutoff},
  journal = {SciPost Phys.},
  volume = {9},
  number = {2},
  pages = {023},
  year = {2020},
  doi = {10.21468/SciPostPhys.9.2.023},
  eprint = {2004.07242},
  archivePrefix = {arXiv},
}

@misc{StanfordYang2020,
  author = {Stanford, Douglas and Yang, Zhenbin},
  title = {Finite-cutoff {JT} gravity and self-avoiding loops},
  year = {2020},
  eprint = {2004.08005},
  archivePrefix = {arXiv},
  primaryClass = {hep-th},
}

@article{GriguoloPaneraiPapaliniSeminara2022,
  author = {Griguolo, Luca and Panerai, Rodolfo and Papalini, Jacopo and Seminara, Domenico},
  title = {Nonperturbative effects and resurgence in {Jackiw--Teitelboim} gravity at finite cutoff},
  journal = {Phys. Rev. D},
  volume = {105},
  number = {4},
  pages = {046015},
  year = {2022},
  doi = {10.1103/PhysRevD.105.046015},
  eprint = {2106.01375},
  archivePrefix = {arXiv},
}

@article{EbertFerkoSunSun2022,
  author = {Ebert, Stephen and Ferko, Christian and Sun, Hao-Yu and Sun, Zhengdi},
  title = {{$T\overline{T}$} in {JT} gravity and {BF} gauge theory},
  journal = {SciPost Phys.},
  volume = {13},
  number = {4},
  pages = {096},
  year = {2022},
  doi = {10.21468/SciPostPhys.13.4.096},
  eprint = {2205.07817},
  archivePrefix = {arXiv},
}

@article{GoelIliesiuKruthoffYang2021,
  author = {Goel, Akash and Iliesiu, Luca V. and Kruthoff, Jorrit and Yang, Zhenbin},
  title = {Classifying boundary conditions in {JT} gravity: from energy-branes to {$\alpha$}-branes},
  journal = {JHEP},
  volume = {2021},
  number = {04},
  pages = {069},
  year = {2021},
  doi = {10.1007/JHEP04(2021)069},
  eprint = {2010.12592},
  archivePrefix = {arXiv},
}

@article{GriguoloPapaliniRussoSeminara2026,
  author = {Griguolo, Luca and Papalini, Jacopo and Russo, Lorenzo and Seminara, Domenico},
  title = {A new perspective on dilaton gravity at finite cutoff},
  journal = {JHEP},
  volume = {2026},
  number = {08},
  pages = {197},
  year = {2026},
  doi = {10.1007/JHEP08(2026)197},
  eprint = {2512.21774},
  archivePrefix = {arXiv},
}

@misc{GriguoloEtAl2026Box,
  author = {Griguolo, Luca and Papalini, Jacopo and Russo, Lorenzo and Seminara, Domenico and Tarana, Alex},
  title = {Quantum {JT} Gravity in a box as a {P{\"o}schl--Teller} Scattering Problem},
  year = {2026},
  eprint = {2607.01385},
  archivePrefix = {arXiv},
  primaryClass = {hep-th},
}

@article{Ferrari2025RandomDisks,
  author = {Ferrari, Frank},
  title = {{Jackiw--Teitelboim} gravity, random disks of constant curvature, self-overlapping curves, and {Liouville CFT1}},
  journal = {Phys. Rev. D},
  volume = {111},
  number = {6},
  pages = {L061901},
  year = {2025},
  doi = {10.1103/PhysRevD.111.L061901},
  eprint = {2402.08052},
  archivePrefix = {arXiv},
}

@article{ChaudhuriFerrari2025OneLoop,
  author = {Chaudhuri, Soumyadeep and Ferrari, Frank},
  title = {Finite cut-off {JT} and {Liouville} quantum gravities on the disk at one loop},
  journal = {JHEP},
  volume = {2025},
  number = {01},
  pages = {160},
  year = {2025},
  doi = {10.1007/JHEP01(2025)160},
  eprint = {2404.03748},
  archivePrefix = {arXiv},
}

@article{Ferrari2026Lattice,
  author = {Ferrari, Frank},
  title = {Random disks of constant curvature: The lattice story},
  journal = {Phys. Rev. D},
  volume = {113},
  number = {2},
  pages = {026012},
  year = {2026},
  doi = {10.1103/3gb4-y8r8},
  eprint = {2406.06875},
  archivePrefix = {arXiv},
}

@article{Ferrari2021Gauge,
  author = {Ferrari, Frank},
  title = {Gauge theory formulation of hyperbolic gravity},
  journal = {JHEP},
  volume = {2021},
  number = {03},
  pages = {046},
  year = {2021},
  doi = {10.1007/JHEP03(2021)046},
  eprint = {2011.02108},
  archivePrefix = {arXiv},
}

@article{MoitraSakeTrivedi2021,
  author = {Moitra, Upamanyu and Sake, Sunil Kumar and Trivedi, Sandip P.},
  title = {{Jackiw--Teitelboim} gravity in the second order formalism},
  journal = {JHEP},
  volume = {2021},
  number = {10},
  pages = {204},
  year = {2021},
  doi = {10.1007/JHEP10(2021)204},
  eprint = {2101.00596},
  archivePrefix = {arXiv},
}

@article{BhattacharyyaEtAl2026,
  author = {Bhattacharyya, Arpan and Ghosh, Saptaswa and Pal, Sounak and Vinod, Anandu},
  title = {Finite Cutoff {Jackiw--Teitelboim (JT)} Gravity: Baby Universes, Matrix Dual, and ({Krylov}) Complexity},
  journal = {Prog. Theor. Exp. Phys.},
  volume = {2026},
  number = {7},
  pages = {073B03},
  year = {2026},
  doi = {10.1093/ptep/ptag083},
  eprint = {2502.13208},
  archivePrefix = {arXiv},
}

@article{Case1950,
  author = {Case, K. M.},
  title = {Singular Potentials},
  journal = {Phys. Rev.},
  volume = {80},
  number = {5},
  pages = {797--806},
  year = {1950},
  doi = {10.1103/PhysRev.80.797},
}

@article{DerezinskiRichard2017,
  author = {Derezi{\'n}ski, Jan and Richard, Serge},
  title = {On {Schr{\"o}dinger} Operators with Inverse Square Potentials on the Half-Line},
  journal = {Ann. Henri Poincar{\'e}},
  volume = {18},
  number = {3},
  pages = {869--928},
  year = {2017},
  doi = {10.1007/s00023-016-0520-7},
  eprint = {1604.03340},
  archivePrefix = {arXiv},
}

@article{CamaraDaSilvaPereiraLima2024,
  author = {Camara da Silva, Ulysses and Pereira, Carlos F. S. and Lima, Andre Alves},
  title = {Renormalization group and spectra of the generalized {P{\"o}schl--Teller} potential},
  journal = {Ann. Phys.},
  volume = {460},
  pages = {169549},
  year = {2024},
  doi = {10.1016/j.aop.2023.169549},
  eprint = {2308.04596},
  archivePrefix = {arXiv},
}

@article{InoueRichard2025,
  author = {Inoue, Hideki and Richard, Serge},
  title = {Scattering theory and an index theorem on the radial part of {$\mathrm{SL}(2,\mathbb{R})$}},
  journal = {J. Topol. Anal.},
  volume = {17},
  number = {5},
  pages = {1445--1479},
  year = {2025},
  doi = {10.1142/S179352532450002X},
  eprint = {2210.02609},
  archivePrefix = {arXiv},
}

@article{DerezinskiKarimi2026,
  author = {Derezi{\'n}ski, Jan and Karimi, Pedram},
  title = {Exactly Solvable {Schr{\"o}dinger} Operators Related to the Hypergeometric Equation},
  journal = {Milan J. Math.},
  volume = {94},
  number = {1},
  pages = {129--193},
  year = {2026},
  doi = {10.1007/s00032-026-00431-9},
  eprint = {2509.03235},
  archivePrefix = {arXiv},
}

@article{GitmanTyutinVoronov2010,
  author = {Gitman, D. M. and Tyutin, I. V. and Voronov, B. L.},
  title = {Self-adjoint extensions and spectral analysis in {Calogero} problem},
  journal = {J. Phys. A: Math. Theor.},
  volume = {43},
  number = {14},
  pages = {145205},
  year = {2010},
  doi = {10.1088/1751-8113/43/14/145205},
  eprint = {0903.5277},
  archivePrefix = {arXiv},
}

@article{CamblongEtAl2000,
  author = {Camblong, Horacio E. and Epele, Luis N. and Fanchiotti, Huner and Garc{\'i}a Canal, Carlos A.},
  title = {Renormalization of the Inverse Square Potential},
  journal = {Phys. Rev. Lett.},
  volume = {85},
  number = {8},
  pages = {1590--1593},
  year = {2000},
  doi = {10.1103/PhysRevLett.85.1590},
  eprint = {hep-th/0003014},
  archivePrefix = {arXiv},
}

@article{BouazizBawin2014,
  author = {Bouaziz, Djamil and Bawin, Michel},
  title = {Singular inverse-square potential: renormalization and self-adjoint extensions for medium to weak coupling},
  journal = {Phys. Rev. A},
  volume = {89},
  number = {2},
  pages = {022113},
  year = {2014},
  doi = {10.1103/PhysRevA.89.022113},
  eprint = {1402.5325},
  archivePrefix = {arXiv},
}

@book{BehrndtHassiDeSnoo2020,
  author = {Behrndt, Jussi and Hassi, Seppo and de Snoo, Henk},
  title = {Boundary Value Problems, {Weyl} Functions, and Differential Operators},
  volume = {108},
  series = {Monographs in Mathematics},
  publisher = {Birkh{\"a}user},
  address = {Cham},
  year = {2020},
  doi = {10.1007/978-3-030-36714-5},
}

@article{KostenkoSakhnovichTeschl2012,
  author = {Kostenko, Aleksey and Sakhnovich, Alexander and Teschl, Gerald},
  title = {{Weyl--Titchmarsh} Theory for {Schr{\"o}dinger} Operators with Strongly Singular Potentials},
  journal = {Int. Math. Res. Not.},
  volume = {2012},
  number = {8},
  pages = {1699--1747},
  year = {2012},
  doi = {10.1093/imrn/rnr065},
  eprint = {1007.0136},
  archivePrefix = {arXiv},
}

@book{Kato1995,
  author = {Kato, Tosio},
  title = {Perturbation Theory for Linear Operators},
  edition = {2},
  series = {Classics in Mathematics},
  publisher = {Springer},
  address = {Berlin, Heidelberg},
  year = {1995},
  doi = {10.1007/978-3-642-66282-9},
  note = {Reprint of the 1976 second edition},
}

@book{Teschl2014,
  author = {Teschl, Gerald},
  title = {Mathematical Methods in Quantum Mechanics: With Applications to {Schr{\"o}dinger} Operators},
  volume = {157},
  edition = {2},
  series = {Graduate Studies in Mathematics},
  publisher = {American Mathematical Society},
  address = {Providence, RI},
  year = {2014},
  doi = {10.1090/gsm/157},
}

@book{SchillingSongVondracek2012,
  author = {Schilling, Ren{\'e} L. and Song, Renming and Vondra{\v c}ek, Zoran},
  title = {{Bernstein} Functions: Theory and Applications},
  volume = {37},
  edition = {2},
  series = {De Gruyter Studies in Mathematics},
  publisher = {Walter de Gruyter},
  address = {Berlin, Boston},
  year = {2012},
  doi = {10.1515/9783110269338},
}

@misc{NISTDLMF2026,
  key = {DLMF},
  title = {{NIST Digital Library of Mathematical Functions}},
  year = {2026},
  howpublished = {\url{https://dlmf.nist.gov/}, Release 1.2.8 of 2026-09-15},
  url = {https://dlmf.nist.gov/},
  note = {{F. W. J. Olver, A. B. {Olde Daalhuis}, D. W. Lozier, B. I. Schneider, R. F. Boisvert, C. W. Clark, B. R. Miller, B. V. Saunders, H. S. Cohl, and M. A. McClain, eds.}},
}

@book{DyatlovZworski2019,
  author = {Dyatlov, Semyon and Zworski, Maciej},
  title = {Mathematical Theory of Scattering Resonances},
  volume = {200},
  series = {Graduate Studies in Mathematics},
  publisher = {American Mathematical Society},
  address = {Providence, RI},
  year = {2019},
  doi = {10.1090/gsm/200},
}

@article{FabrisRicharteSaa2021,
  author = {Fabris, J{\'u}lio C. and Richarte, Mart{\'i}n G. and Saa, Alberto},
  title = {Quasinormal modes and self-adjoint extensions of the {Schr{\"o}dinger} operator},
  journal = {Phys. Rev. D},
  volume = {103},
  number = {4},
  pages = {045001},
  year = {2021},
  doi = {10.1103/PhysRevD.103.045001},
  eprint = {2010.10674},
  archivePrefix = {arXiv},
}

@misc{SheikhJabbariTaghiloo2026,
  author = {Sheikh-Jabbari, M. M. and Taghiloo, V.},
  title = {{GR} from {RG}, {$2d$} Example: {JT}-Gravity Induced from Renormalization Group Flow},
  year = {2026},
  eprint = {2605.28753},
  archivePrefix = {arXiv},
  primaryClass = {hep-th},
  doi = {10.1103/7yr1-bd1m},
  note = {Accepted in Phys. Rev. D},
}

\end{document}